\documentclass[a4paper,onecolumn,accepted=2026-09-11]{quantumarticle}
\pdfoutput=1
\usepackage{latexsym,amssymb,amsmath,amsfonts,amscd,epsfig,amsthm,xcolor,mathrsfs,braket,mleftright,mathtools}
\usepackage{tikz}
\usepackage{graphicx}
\usepackage{pgfplots}
\pgfplotsset{compat=1.18}
\usepackage[english]{babel}
\usepackage[left=2.3cm,right=2.3cm,top=2.3cm,bottom=2.3cm]{geometry}
\usepackage[bookmarksnumbered,linktocpage,hypertexnames=false,colorlinks=true,linkcolor=blue,urlcolor=blue,citecolor=blue,anchorcolor=green,breaklinks=true,pdfusetitle]{hyperref}
\usepackage[capitalize,nameinlink]{cleveref}
\usepackage{aliascnt}
\usepackage{authblk}
\usepackage[numbers,sort&compress]{natbib}

\newtheorem{thm}{Theorem}\crefname{thm}{Theorem}{Theorems}
\crefname{lem}{Lemma}{Lemmas}
\crefname{prop}{Proposition}{Propositions}
\crefname{cor}{Corollary}{Corollaries}
\newtheorem*{thm*}{Theorem}
\theoremstyle{definition}
\newtheorem{dfn}[thm]{Definition}\crefname{def}{Definition}{Definitions}
\crefname{exa}{Example}{Examples}
\crefname{rmk}{Remark}{Remarks}

\DeclareMathOperator{\tr}{tr}

\newcommand{\QMA}{{\sf QMA}}
\newcommand{\NP}{{\sf NP}}
\newcommand{\PP}{{\sf PP}}
\newcommand{\PostBQP}{{\sf PostBQP}}
\newcommand{\PSPACE}{{\sf PSPACE}}
\newcommand{\PreciseQMA}{{\sf PreciseQMA}}

\DeclareMathOperator{\poly}{poly}
\DeclareMathOperator{\acc}{acc}
\DeclareMathOperator{\rej}{rej}

\DeclarePairedDelimiter{\abs}{\lvert}{\rvert}

\title{Limits to black-box amplification in QMA}

\author[1]{Scott Aaronson}
\author[2]{Phillip Harris}
\author[3]{Freek Witteveen}
\affil[1]{University of Texas at Austin. Supported by the Simons Foundation}
\affil[2]{University of Bonn}
\affil[3]{QuSoft \& CWI, Amsterdam}
\date{}

\begin{document}

\maketitle

\begin{abstract}
We study the limitations of black-box amplification in the quantum complexity class $\QMA$. 
Amplification is known to boost any inverse-polynomial gap between completeness and soundness to exponentially small error, and a recent result (Jeffery and Witteveen, 2025) shows that completeness can in fact be amplified to be doubly exponentially close to~1. 
We prove that this is optimal for black-box procedures: we provide a quantum oracle relative to which no $\QMA$ verification procedure using polynomial resources can achieve completeness closer to~1 than doubly exponential, or a soundness which is super-exponentially small.
This is proven by making the oracle separation from (Aaronson, 2009) between $\QMA$ and $\QMA_1$ quantitative, using techniques from complex approximation theory.
\end{abstract}

\section{Introduction}

The complexity class $\QMA$ (Quantum Merlin--Arthur) is a quantum analogue of $\NP$, where given a problem instance, a prover Merlin sends a quantum witness to a polynomial-time quantum verifier Arthur, in order to convince Arthur whether the problem is a yes-instance or a no-instance.
Arthur then performs a polynomial-size quantum computation on the witness, after which he accepts or rejects the witness.
Quantum complexity classes typically allow some probability of error.
There are two parameters quantifying the allowed error: the \emph{completeness} $c$, the probability with which Arthur accepts a valid witness in the yes-case, and the \emph{soundness} $s$, the maximum acceptance probability in the no-case.
It is well known that the precise choice of these parameters does not matter as long as there is a non-negligible gap: by amplification techniques \cite{kitaev2002classical,marriott2005quantum}, one can boost a polynomially small gap $c - s = 1 / \poly$ to completeness $c = 1 - 2^{-\poly}$ and soundness $s = 2^{-\poly}$ exponentially close to 1 and 0 respectively.
This motivates the canonical definition, which takes $c = \frac23$ and $s = \frac13$.

A long-standing open question is whether the completeness can in fact be made \emph{perfect}.
We denote by $\QMA_1$ the variant of $\QMA$ where we take $c = 1$ (and $s = \frac13$, or any other constant).
The question whether $\QMA$ equals $\QMA_1$ is then: can every QMA protocol be modified so that Arthur always accepts a valid witness in the yes-case, while still rejecting no-instances with bounded probability?
There are a number of closely related complexity classes which allow for perfect completeness. This is the case for the variant with only classical randomness ${\sf MA}$, the variant with a classical proof and quantum verifier ${\sf QCMA}$ \cite{jordan2012QCMAisQCMA1}, the variant with multiple rounds ${\sf QIP}$ \cite{marriott2005quantum}, and the variant with an exponentially small soundness-completeness gap $\PreciseQMA$ \cite{fefferman2018complete}.

For $\QMA$, this question has resisted resolution, and one explanation for this is given by
Aaronson \cite{aaronson2009perfect}, who gave a barrier to proving $\QMA = \QMA_1$ using black-box techniques.
Here `black-box' refers to an amplification procedure that does not use any properties of the verification circuit, other than those that define it.
Concretely, \cite{aaronson2009perfect} proved that there exists a quantum oracle relative to which $\QMA \neq \QMA_1$. Any proof showing $\QMA = \QMA_1$ would have to break down in the presence of a quantum oracle, i.e. it would have to be quantumly nonrelativizing.
All known amplification techniques for ${\sf QMA}$ are of a black-box nature (so they are quantumly relativizing).

Recently, it was shown in Ref~\cite{jeffery2026infinitecounter} that one can reach completeness \emph{doubly exponentially} close to 1
\begin{align*}
    c = 1 - 2^{-2^{\poly}},
\end{align*}
in $\QMA$, improving over the previous best completeness which was exponentially close to 1. 
The techniques in~\cite{jeffery2026infinitecounter} are of a black-box nature. This raises the question: can one get even closer to perfect completeness with black-box amplification procedures?

In this work we show that the doubly-exponential bound is in fact \emph{optimal} in the black-box setting. 
We prove that there exists a quantum oracle relative to which no $\QMA$ protocol using polynomial resources can achieve completeness closer to~1 than doubly exponential.

The key idea in the oracle separation of \cite{aaronson2009perfect} is as follows.
The maximal acceptance probability is the largest eigenvalue of the accepting measurement operator.
For a simple choice of quantum oracle, the eigenvalues of the accepting measurement operator are analytic functions in one variable.
Perfect completeness would require an eigenvalue to be constantly 1 on an interval, but the analyticity then implies that the maximal acceptance probability is \emph{always} 1, contradicting soundness.
We use the same quantum oracle, but give a slightly different analysis and use a standard result bounding the growth of trigonometric polynomials, to make this separation quantitative.
The key idea of the proof is as follows.
The quantum oracle is given by
\begin{align}\label{eq:oracle intro}
    U(\theta) = \begin{pmatrix}
        \cos(\theta) & \sin(\theta) \\ - \sin(\theta) & \cos(\theta)
    \end{pmatrix},
\end{align}
for $\theta \in [-\pi, \pi)$.
We then consider the problem where Arthur has to decide whether $\abs{\theta} \leq \pi - s$ (yes-instance) or $\theta = \pi$ (no-instance), with black-box access to $U(\theta)$.

\begin{figure}[t]
\begin{tikzpicture}
    \begin{axis}[
            axis lines=left,
            xmin=0, xmax=pi + 0.1,
            ymin=0, ymax=1.05,
            width=12cm, height=6.4cm,
            xtick={0,3*pi/4,pi},
            xticklabels={$0$,$3\pi/4$,$\pi$},
            ytick={0,1/3,1},
            yticklabels={$0$,$\frac13$,$1$},
            xlabel={$\theta$},
            ylabel={$p$},
            samples=400,
        ]

        % guide lines
        \addplot [densely dashed, black] {1};
        \addplot [densely dashed, red] {1 - 0.03};
        \addplot [densely dashed, red]   {1/3};
        \addplot [densely dashed, blue, mark=none] coordinates {(3*pi/4, 0) (3*pi/4, 1.05)};
        \addplot [densely dashed, blue, mark=none] coordinates {(pi, 0) (pi, 1.05)};

        % function acceptance probability
        \addplot [very thick, samples=450, domain=0:pi]
        { 0.97/(1 + exp(-20*(3*pi/4 + 0.25 - x)))
            + (1 - 1/(1 + exp(-20*(7*pi/8 - x)))) * (1/3)*(cos(deg(8*x)))^2
            + 0.03*(cos(deg(24*x)))^2 / (1 + exp(-15*(3*pi/4 + 0.05 - x))) };
    \end{axis}
\end{tikzpicture}
\caption{The verifier would like to decide whether $\abs{\theta} \leq 3\pi/4$ (a yes-instance) or $\theta = \pi$ (a no-instance), given black-box access to the oracle defined in \cref{eq:oracle intro}. Here we give a sketch of what the maximal acceptance probability, optimized over all choices of witness, should look like, when we have completeness $c = 1 - \delta$ close to 1, and soundness $s = 1/3$.}
\label{fig:acceptance probability}
\end{figure}
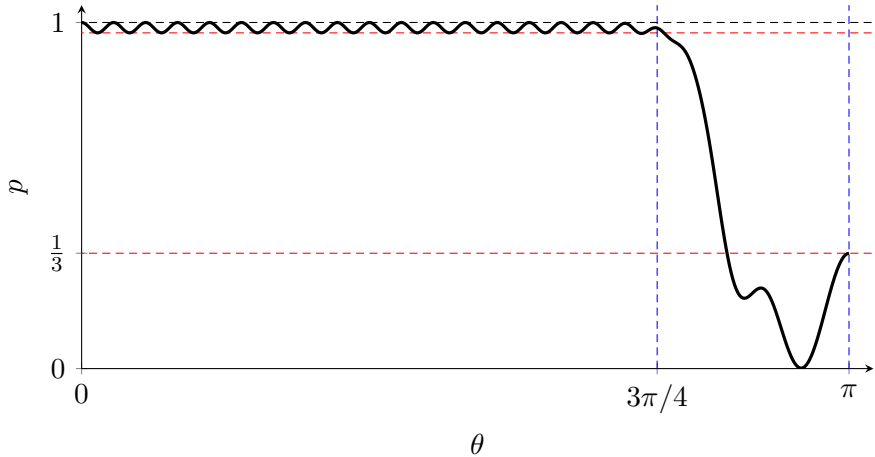

Suppose we are given some verifier for this task.
The idea is to consider the accepting measurement operator $P_{\acc}(\theta)$ and the rejecting measurement operator $P_{\rej}(\theta) = I - P_{\acc}(\theta)$.
The maximal acceptance probability $p(\theta)$, over all possible witnesses, equals the largest eigenvalue of $P_{\acc}(\theta)$.
If the verifier has completeness close to 1, and soundness $\frac13$, the maximal acceptance probability, as a function of the angle $\theta$, must look something like shown in \cref{fig:acceptance probability}.
We now study the following function:
\begin{align}\label{eq:determinant}
    p(\theta) = \det[P_{\rej}(\theta)].
\end{align}
If the verifier has completeness close to 1, $P_{\rej}(\theta)$ has an eigenvalue close to 0, and hence $p(\theta)$ must be small on the accepting values of $
\theta$. On the other hand, we can use soundness to bound how small $p(\pi)$ can be.
We can then use the fact that $p(\theta)$ is a trigonometric polynomial, and a standard result on the growth of polynomials, to bound how close to 1 the completeness can be.
As a side comment, we note that this also simplifies the argument of \cite{aaronson2009perfect}: if one has perfect completeness, $p(\theta) = 0$ on an interval, but then $p(\theta)$, being a trigonometric polynomial and hence analytic, has to be identically zero for all $\theta$, implying that the soundness condition can not be satisfied.

We also investigate the analogous question for the soundness parameter in the definition of $\QMA$.
It is known that $\QMA$ with perfect soundness (and exponentially small completeness error) equals the complexity class $\sf{NQP}$ \cite{kobayashi2003quantum}, which is hard for the polynomial hierarchy \cite{fenner1999determining}, and thus likely different from $\QMA$.
The soundness can be amplified to be exponentially small with standard techniques.
Can one do even better? Can one get, analogous to the situation for completeness, to doubly-exponentially small soundness?
We show, using the same oracle but with the role of the yes- and no-instances reversed, and using
\begin{align}
    p(\theta) = \tr[P_{\acc}(\theta)]
\end{align}
that there is no black-box way to do so.
This completely settles what completeness and soundness parameters are achievable with black-box amplification techniques in $\QMA$.

A caveat is that the proof breaks down if the witness register is allowed to be infinite dimensional, as also observed in \cite{aaronson2009perfect}. This is not just a quirk of the proof technique: in \cite{jeffery2026infinitecounter} it is proven that (with black-box amplification techniques) one in fact can prove perfect completeness for $\QMA$, provided that one allows computation on an infinite-dimensional Hilbert space with an appropriate gate set.
This choice of gate set is such that the infinite-dimensional Hilbert space does not increase the computational power of ${\sf BQP}$ or ${\sf QMA}$.

\paragraph{Author contributions and use of AI:}
The first version of this paper was written by SA and FW, who contributed equally.
In this first version of this paper, we used a different function in the completeness case, namely the \emph{rational} trigonometric function $r(\theta) = \tr[P_{\rej}(\theta)^{-1}]$, which was suggested to us by GPT-5-Thinking.  To analyze this $r$, we used a bound on the growth of complex rational functions \cite{gonvcar1969zolotarev}, which in turn required the no-instances to be an interval instead of a single point. See \cite{aaronson2025shtetl} for a discussion and comments. After the initial posting, one of us (PH, subsequently added as coauthor) suggested to replace $\tr[P_{\rej}(\theta)^{-1}]$ by the simpler $\det[P_{\rej}(\theta)]$.

\section{Preliminaries}

We first recall the definition of ${\sf QMA}$.

\begin{dfn}\label{dfn:qma}
    For parameters $c,s \in [0,1]$ with $c > s$, the class $\QMA_{c,s}$ consists of promise problems $L = (L_{\text{yes}}, L_{\text{no}})$ for which there exists a uniform family of polynomial-time quantum circuits $\{V_x\}_{x \in \{0,1\}^n}$, called verifiers, with the following properties:
    \begin{itemize}
        \item (Completeness) If $x \in L_{\text{yes}}$, then there exists a quantum witness state $\ket{\psi}$ such that
        \begin{align*}
            \Pr[V_x \text{ accepts } \ket{\psi}] \geq c.
        \end{align*}
        \item (Soundness) If $x \in L_{\text{no}}$, then for all quantum states $\ket{\psi}$,
        \begin{align*}
            \Pr[V_x \text{ accepts } \ket{\psi}] \leq s.
        \end{align*}
    \end{itemize}
    The class $\QMA$ is defined as $\QMA_{2/3,1/3}$, since any inverse-polynomial gap between $c$ and $s$ can be amplified to these parameters by standard techniques \cite{marriott2005quantum}. 
    The subclass $\QMA_1$ is defined as $\QMA_{1,s}$ for some constant $s < 1$, i.e., \emph{QMA with perfect completeness}.
\end{dfn}

We will also consider $\QMA$ with access to a quantum oracle $\mathcal U$.
Here, there is a collection of possible unitaries, and the verifier gets access to a unitary $U$ from this set.
The verifier has black-box access to $U$, its inverse, and to their controlled versions.
We then denote by $\QMA^{\mathcal U}$ the class of problems that can be solved as in \cref{dfn:qma} where Arthur additionally has access to a polynomial number of calls to $\mathcal U$.

\subsection{Trigonometric polynomials}

% We use analytic techniques from approximation theory. 
A \emph{trigonometric polynomial} of degree $d$ is a polynomial of degree $d$ in $\cos(\theta)$ and $\sin(\theta)$. 

We use a standard bound on how fast such polynomials can grow: if $p(\theta)$ is uniformly bounded on some interval, how large can it be outside the interval?
The solution is similar to the case of standard polynomials, where the extremal polynomials are Chebyshev polynomials.

We use the following bound; see Theorem 5.1.2 in \cite{borwein2012polynomials}:

\begin{thm}\label{thm:growth bound}
    Let $p(\theta)$ be a trigonometric polynomial of degree $d$. Then for $u \in (0,\pi/4]$,
    \begin{equation}
        \max_{\theta} \abs{p(\theta)} \leq \exp(8du) \max_{\abs{\theta} \leq \pi - u} \abs{p(\theta)}
    \end{equation}
\end{thm}

That is, if $p(\theta)$ is uniformly bounded on a subset of angles, we can bound the maximal value it can take.

\section{Limits to black-box QMA amplification}

In this section we prove tight black-box lower bounds for amplifying the completeness and soundness parameters of $\QMA$.
The following result shows that there is an oracle with respect to which $\QMA$ has completeness $c(n)$ which is at most doubly-exponentially close to 1, and soundness which is at most exponentially small.
Both bounds follow from the same oracle construction (as in \cite{aaronson2009perfect}) together with the polynomial growth bound in \cref{thm:growth bound}.

\begin{thm}[Limits to black-box amplification]\label{thm:black-box-limits}
There exists a quantum oracle $\mathcal U$ such that the following hold for any $\QMA$ verification procedure using $\poly(n)$ resources:
\begin{enumerate}
    \item\label{it:completeness} \textbf{(Completeness barrier).} For any black-box $\QMA$ amplification procedure that achieves completeness $c = 1-\delta$ and soundness $s = 1/3$, we have
    \begin{align*}
        \delta \geq 2^{-2^{\poly(n)}}.
    \end{align*}
    More precisely, relative to $\mathcal U$,
    \begin{align*}
        \QMA^{\mathcal U} \neq \QMA_{c(n),\,1/3}^{\mathcal U}
        \qquad \text{for } c(n)=1 - 2^{-2^{n^{\omega(1)}}}. %o\mleft(2^{-2^{\poly(n)}}\mright).
    \end{align*}
    \item\label{it:soundness} \textbf{(Soundness barrier).} For any black-box $\QMA$ amplification procedure that achieves completeness $c = 2/3$ and soundness $s=\delta$, we have
    \begin{align*}
        \delta \geq 2^{-\poly(n)}.
    \end{align*}
    More precisely, relative to $\mathcal U$,
    \begin{align*}
        \QMA^{\mathcal U} \neq \QMA_{2/3,\,s(n)}^{\mathcal U}
        \qquad \text{for } s(n) = 2^{-n^{\omega(1)}}. %o\mleft(2^{-\poly(n)}\mright).
    \end{align*}
\end{enumerate}
\end{thm}

\begin{proof}
We use the same quantum oracle as in \cite{aaronson2009perfect}. 
For $\theta \in (-\pi,\pi]$, define
\begin{align*}
U(\theta)=
\begin{pmatrix}
\cos\theta & \sin\theta\\
-\sin\theta & \cos\theta
\end{pmatrix}.
\end{align*}

The promise problem is to decide whether the oracle parameter satisfies $\abs{\theta} \leq \pi - u$ for some arbitrary constant $u \in (0,\pi/4]$, or $\theta = \pi$.
Which of these two cases corresponds to yes-instances and no-instances is reversed for the proof of the completeness and soundness barrier, but in either case it is obvious that Arthur can do so with constant soundness and completeness (and even with trivial witness), so the problem is in $\QMA$.
We will now show, however, that there are limits to how close the completeness and soundness can get to 1 and 0 respectively using polynomial resources.

Consider any $\QMA$ verifier $V$ that, on input $x\in\{0,1\}^n$, uses
\begin{itemize}
    \item $t=t(n)=\poly(n)$ calls to the oracle,
    \item a witness register of $m=m(n)=\poly(n)$ qubits; let $q=2^{m}$ denote its dimension.
\end{itemize}
Let $P_{\acc}(\theta)$ denote the POVM element corresponding to acceptance and let $P_{\rej}(\theta) = I - P_{\acc}(\theta)$.
Write the eigenvalues of $P_{\acc}(\theta)$ as $1 \geq \lambda_1(\theta) \geq \cdots \geq \lambda_q(\theta) \geq 0$. 
For each fixed $\theta$, the optimal acceptance probability, over all choices of witness, equals $\lambda_1(\theta)$.

Every matrix element of $U(\theta)$ and $U(\theta)^\dagger$  (and hence also of their controlled versions) is affine in $\cos(\theta)$ and $\sin(\theta)$.
Hence, by expanding the verifier circuit and projecting on the accept outcome, each entry of $P_{\acc}(\theta)$ and $P_{\rej}(\theta)$ is a trigonometric polynomial of degree $2t$ in $\theta$.

\paragraph{Completeness barrier}

We let $\abs{\theta} \leq \pi - u$ be the yes-instances, and $\theta = \pi$ be the no-instance.
Assume the verifier achieves completeness $1-\delta$ and soundness $1/3$:
\begin{align*}
\lambda_1(\theta)\geq 1-\delta \text{ for yes-instances},
\qquad
\lambda_1(\theta)\leq \tfrac{1}{3} \text{ for a no-instance.}
\end{align*}
We consider the following function:
\begin{align*}
p(\theta)=\det\bigl[P_{\rej}(\theta) \bigr]=\prod_{i=1}^q (1-\lambda_i(\theta)) \, ,
\end{align*}
which is a trigonometric polynomial of degree $2tq$.

We now bound $p$ for yes- and no-instances.
For yes-instances, $1 - \lambda_1(\theta) \leq \delta$, so $p(\theta) \leq \delta$.
For the no-instance, all $\lambda_i(\theta) 
\leq 1/3$ and hence $p(\pi) \geq (2/3)^q$.

By \cref{thm:growth bound},
\begin{align*}
    (2/3)^q \leq \max_{\theta} \abs{p(\theta)} \leq \exp(16utq) \max_{\abs{\theta} \leq \pi- u} \abs{p(\theta)} \leq \exp(16utq) \delta.
\end{align*}
In other words,
\begin{align}
    \delta \geq (2/3)^q \exp(-16utq)
\end{align}
and since $q = 2^{\poly(n)}$, and $t = \poly(n)$, this proves \ref{it:completeness}.

\paragraph{Soundness barrier}
We let $\abs{\theta} \leq \pi - u$ be the no-instances, and $\theta = \pi$ be the yes-instance.
Assume the verifier achieves completeness $2/3$ and soundness $\delta$, so
\begin{align*}
    \lambda_1(\theta) &\leq \delta \text{ for no-instances}\\
    \lambda_1(\theta) &\geq \frac23 \text{ for a yes-instance.}
\end{align*}
The idea is now very similar to the completeness case, but using a different polynomial
\begin{align*}
    p(\theta) = \tr[P_{\acc}(\theta)] = \sum_{i=1}^q \lambda_i(\theta).
\end{align*}
This is a trigonometric polynomial of degree $2t$ in $\theta$.
We can bound it as
\begin{align*}
    p(\theta) &\leq q \lambda_1(\theta) \leq q \delta \text{ for no-instances,}\\
    p(\theta) &\geq \lambda_1(\theta) \geq \frac{2}{3} \text{ for a yes-instance.}
\end{align*}
By \cref{thm:growth bound},
\begin{align*}
    \frac23 \leq \max_{\theta} \abs{p(\theta)} \leq \exp(16ut) \max_{\abs{\theta} \leq \pi- u} \abs{p(\theta)} \leq \exp(16ut) q \delta.
\end{align*}
In other words,
\begin{align}
    \delta \geq \frac{2}{3q} \exp(-16ut)
\end{align}
and since $q = 2^{\poly(n)}$, and $t = \poly(n)$, this proves \ref{it:soundness}.

\end{proof}

Note that in the proof of \cref{thm:black-box-limits}, the crucial difference between the completeness and the soundness cases is that the degree of the function $p(\theta)$ does not depend on the witness dimension $q$ for the soundness.

Since black-box amplification procedures can also \emph{achieve} completeness $c = 1 - 2^{-2^{\poly(n)}}$ and soundness $s = 2^{-\poly(n)}$, this gives the optimal completeness and soundness parameters for $\QMA$ which can be achieved through black-box reductions.

\section{Discussion}
While previously it was known how to achieve completeness and soundness exponentially close to 1 and 0 respectively, this work, together with \cite{jeffery2026infinitecounter}, shows that with black-box reductions one can achieve completeness doubly exponentially close to 1, and soundness exponentially small, and that this is optimal. The asymmetry between soundness and completeness is natural from the definition: completeness only requires a single eigenvalue of the accepting measurement operator to be very close to 1, while soundness is a condition on all possible witness states and requires \emph{all} eigenvalues of the accepting measurement operator to be close to 0.

Let us call attention to something conceptually strange about our result.  Recall that $\QMA \subseteq \PP = \PostBQP$, with the containment believed to be strict. Nevertheless, we saw in this paper that ``the approximation theory object corresponding to $\QMA$,'' namely the largest eigenvalue of an $\exp(n)\times \exp(n)$ Hermitian matrix of low-degree polynomials, is in some respects \textit{more} powerful than ``the approximation theory object corresponding to $\PP$,'' namely a low-degree rational function. In particular, the former allows amplification to doubly exponentially small completeness error, while the latter does not.

How is this possible? In our view, the resolution is simply that, to prove $\QMA \subseteq \PP$, it suffices to calculate only a crude estimate of the largest eigenvalue in question---something that \textit{is} possible using low-degree polynomials. Indeed, if we ask for more precise information about the largest eigenvalue, we get the class $\PreciseQMA$, which equals $\PSPACE$ \cite{fefferman2018complete} and hence is presumably more powerful than $\PostBQP = \PP$.

Of course, the central open question we leave is whether $\QMA = \QMA_1$ for some, or all, finite-dimensional gate sets.

% A more minor open question is whether, in our oracle separation, we can take the angle $\theta$ to have some fixed value, say $\theta=\pi$, in the no-case, rather than belonging to an interval of values. \ We conjecture that the answer is yes but a new argument seems needed.

\subsection*{Acknowledgments}
We thank Stacey Jeffery for helpful discussions. FW was supported by the European Union (ERC, ASC-Q, 101040624). PH was supported by Deutsche Forschungsgemeinschaft (DFG) Grant No. SFB-TRR 358/1 2023 — 491392403.

\bibliographystyle{quantum}
\bibliography{references}

\newcommand{\arxiv}[1]{arXiv: \href{https://arxiv.org/abs/#1}{\ttfamily{#1}}}
\begin{thebibliography}{10}

\bibitem{kitaev2002classical}
Alexei~Yu Kitaev, Alexander Shen, and Mikhail~N Vyalyi.
\newblock ``Classical and quantum computation''.
\newblock Number~47 in Graduate Studies in Mathematics. American Mathematical
  Society. ~(2002).

\bibitem{marriott2005quantum}
Chris Marriott and John Watrous.
\newblock ``Quantum {Arthur--Merlin} games''.
\newblock
  \href{https://dx.doi.org/https://doi.org/10.1007/s00037-005-0194-x}{Computational
  Complexity {\bf 14}, 122--152}~(2005).
\newblock  \href{http://arxiv.org/abs/cs/0506068}{arXiv:cs/0506068}.

\bibitem{jordan2012QCMAisQCMA1}
Stephen~P. Jordan, Hirotada Kobayashi, Daniel Nagaj, and Harumichi Nishimura.
\newblock ``Achieving perfect completeness in classical-witness quantum
  {Merlin-Arthur} proof systems''.
\newblock
  \href{https://dx.doi.org/https://doi.org/10.26421/QIC12.5-6-7}{Quantum
  Information and Computation {\bf 12}, 461--471}~(2012).
\newblock  \href{http://arxiv.org/abs/1111.5306}{arXiv:1111.5306}.

\bibitem{fefferman2018complete}
Bill Fefferman and Cedric Yen-Yu Lin.
\newblock ``A complete characterization of unitary quantum space''.
\newblock In 9th Innovations in Theoretical Computer Science Conference (ITCS
  2018).
\newblock Pages 4--1.
\newblock Schloss Dagstuhl--Leibniz-Zentrum f{\"u}r Informatik~(2018).
\newblock  \href{http://arxiv.org/abs/1604.01384}{arXiv:1604.01384}.

\bibitem{aaronson2009perfect}
Scott Aaronson.
\newblock ``On perfect completeness for {QMA}''.
\newblock \href{https://dx.doi.org/https://doi.org/10.26421/QIC9.1-2-5}{Quantum
  Information and Computation {\bf 9}, 81--89}~(2009).
\newblock  \href{http://arxiv.org/abs/0806.0450}{arXiv:0806.0450}.

\bibitem{jeffery2026infinitecounter}
Stacey Jeffery and Freek Witteveen.
\newblock ``Quantum-merlin-arthur problems have perfect completeness with an
  infinite counter''.
\newblock \href{https://dx.doi.org/https://doi.org/10.1103/pwdd-htbf}{Physical
  Review Letters {\bf 136}, 180601}~(2026).
\newblock  \href{http://arxiv.org/abs/2506.15551}{arXiv:2506.15551}.

\bibitem{kobayashi2003quantum}
Hirotada Kobayashi, Keiji Matsumoto, and Tomoyuki Yamakami.
\newblock ``Quantum {Merlin-Arthur} proof systems: Are multiple {Merlins} more
  helpful to {Arthur}?''.
\newblock In International Symposium on Algorithms and Computation.
\newblock
  \href{https://dx.doi.org/https://doi.org/10.1007/978-3-540-24587-2_21}{Pages
  189--198}.
\newblock Springer~(2003).
\newblock
  \href{http://arxiv.org/abs/quant-ph/0306051}{arXiv:quant-ph/0306051}.

\bibitem{fenner1999determining}
Stephen Fenner, Frederic Green, Steven Homer, and Randall Pruim.
\newblock ``Determining acceptance possibility for a quantum computation is
  hard for the polynomial hierarchy''.
\newblock
  \href{https://dx.doi.org/https://doi.org/10.1098/rspa.1999.0485}{Proceedings
  of the Royal Society of London. Series A: Mathematical, Physical and
  Engineering Sciences {\bf 455}, 3953--3966}~(1999).
\newblock
  \href{http://arxiv.org/abs/quant-ph/9812056}{arXiv:quant-ph/9812056}.

\bibitem{gonvcar1969zolotarev}
AA~Gon{\v{c}}ar.
\newblock ``Zolotarev problems connected with rational functions''.
\newblock \href{https://dx.doi.org/10.1070/SM1969v007n04ABEH001107}{Mathematics
  of the USSR-Sbornik {\bf 7}, 623}~(1969).

\bibitem{aaronson2025shtetl}
Scott Aaronson.
\newblock ``{Shtetl Optimized: The QMA Singularity}''.
\newblock \url{https://scottaaronson.blog/?p=9183}~(2025).

\bibitem{borwein2012polynomials}
Peter Borwein and Tam{\'a}s Erd{\'e}lyi.
\newblock ``Polynomials and polynomial inequalities''.
\newblock Volume 161.
\newblock Springer Science \& Business Media. ~(2012).

\end{thebibliography}

\end{document}